\documentclass[12pt]{article}

\usepackage[english]{babel}

\usepackage[a4paper]{geometry}
\usepackage{amssymb, amsmath, amsthm}
\usepackage{enumitem, mathtools, bm}
\usepackage{graphicx, cite, color}
\usepackage{longtable, pdflscape, booktabs, caption, multicol}

\usepackage[hyphens]{url}
\usepackage[colorlinks=true,citecolor=black,linkcolor=black,urlcolor=blue]{hyperref}
\usepackage{breakurl}

\theoremstyle{plain}
\newtheorem{theorem}{Theorem}
\newtheorem{lemma}[theorem]{Lemma}
\newtheorem{corollary}[theorem]{Corollary}

\theoremstyle{definition}
\newtheorem{definition}[theorem]{Definition}

\theoremstyle{remark}

\newcommand{\doi}[1]{\url{https://doi.org/#1}}

\def\slovakL{L\kern-0.25em\char39\kern-0.03em}

\DeclareMathOperator{\Tr}{Tr}
\newcommand{\nhalf}{\lfloor\frac{n-1}{2}\rfloor}
\newcommand{\khalf}{\lfloor\frac{k-1}{2}\rfloor}

\usepackage{authblk}

\title{An efficient algorithm for generating\\ transmission irregular trees}

\author[1]{Ivan Stošić}
\author[2,3]{Ivan Damnjanović}
\affil[1]{Faculty of Sciences and Mathematics, University of Niš, Niš, Serbia}
\affil[2]{Faculty of Mathematics, Natural Sciences and Information Technologies, University of Primorska, Koper, Slovenia}
\affil[3]{Faculty of Electronic Engineering, University of Niš, Niš, Serbia}

\date{}

\begin{document}

\maketitle

\begin{abstract}
The transmission of a vertex in a connected graph is the sum of distances from that vertex to all the other vertices. A connected graph is transmission irregular if any two distinct vertices have different transmissions. We present an efficient algorithm that generates all the transmission irregular trees up to a given order, up to isomorphism.
\end{abstract}



\section{Introduction}

Let $G$ be a simple connected graph with the vertex set $V(G)$ and the edge set $E(G)$. The \emph{distance} $d_G(u, v)$ between two vertices $u, v \in V(G)$ is the length of a shortest $(u, v)$-walk in $G$. The \emph{transmission}
\[
    \Tr_G(v) = \sum_{u \in V(G)} d_G(v, u)
\]
of a vertex $v \in V(G)$ is the sum of distances from $v$ to all the other vertices of $G$; see \cite{Abiad2017, RaMaPaASu2016, Soltes1991}. Alternatively, the transmission of a vertex may be referred to as the total distance of a vertex \cite{CaDoSco2019, KlaZe2013} or the status of a vertex \cite{ABriGri2021, QiaoZhan2020}.

The \emph{transmission set} of a connected graph $G$ is
\[
    \Tr(G) = \{ \Tr_G(v) \colon v \in V(G)\}.
\]
The \emph{Wiener complexity} \cite{AliKla2018, XuIIrKlaLi}, or alternatively the Wiener dimension \cite{AliAnKlaSkre2014}, of graph $G$ is defined as $|\Tr(G)|$. A connected graph $G$ is \emph{transmission irregular}, or TI for short, if $|\Tr(G)| = |V(G)|$, i.e., if any two distinct vertices have different transmissions.

As shown by Alizadeh and Klav\v{z}ar \cite[Theorem 2.1]{AliKla2018}, almost all graphs are not TI. For this reason, many graph theorists have found interest in investigating various properties of TI graphs. For instance, several infinite families of TI graphs that satisfy certain additional constraints have been found by Dobrynin \cite{Dobrynin2018, Dobrynin2019_AMC, Dobrynin2019_DAM, Dobrynin2019_DM}, as well as Bezhaev and Dobrynin \cite{BeDo2021, BeDo2022}. Notable work has also been done on the interval TI graphs \cite{AlYaSte2022_JAMC}, which are a subclass of the TI graphs, and the closely related stepwise TI graphs \cite{AliKla2023, AlYaSte2022_AMC, DoSha2020} and generalized stepwise TI graphs \cite{AliKlaMo2024}.

In particular, there are many results on TI trees; see \cite{Dobrynin2019_DM, QiaoZhan2020, XuKla2021, XuTianKla2023}. For instance, the complete characterization of TI starlike trees was recently obtained through a series of papers \cite{AliKla2018, AlYaSte2020, Damnjanovic2024}. The interest in TI trees is natural because, as it turns out, it is easier to inspect whether a given graph is TI when (almost) all edges are a cut edge. Besides, many realizability problems involving TI graphs can be resolved by using either trees, unicyclic or bicyclic graphs; see \cite{XuTianKla2023, DamSteAlYa2024, AXu}.

Here, we present an efficient algorithm for generating all the TI trees up to a given order, up to isomorphism. Although there exist efficient tools capable of finding all the trees of a given order, such as the program \texttt{gentreeg} from the package \texttt{nauty} \cite{McKayPip2014}, there is a lack of built-in support for TI trees. For example, there are $14,830,871,802$ trees on $30$ vertices \cite{Sloane}, while only $302,163$ of them are TI. Thus, it is both time and memory consuming to search for TI trees through the na\"{i}ve method of generating all the trees of the same order and then selecting the TI trees among them. Our goal is to facilitate future research on TI trees by providing a tool that efficiently outputs only the TI trees up to a given order without having to do a full filtration over the set of all the trees.

In Section \ref{sc_properties}, we overview some known results on TI trees and obtain several auxiliary lemmas to be used later. Afterwards, we develop the generation algorithm and provide some implementational details in Section \ref{sc_algorithm}. The developed program implemented in \texttt{C++20} can be found in \cite{Supplementary}, while Section~\ref{sc_results} contains a brief discussion on the execution results of the implemented program.

\section{Preliminaries}\label{sc_properties}

For any graph $G$, we will use $|G|$ and $d_G(v)$ to denote the order of $G$ and the degree of vertex $v \in V(G)$ in $G$, respectively. Also, for any connected graph $G$ and two vertices $u, v \in V(G)$, let $n_G(u, v)$ denote the number of vertices in $G$ that are closer to $u$ than to $v$. With this in mind, we state the following folklore lemma whose proof can be found, e.g., in \cite[Lemma~6]{Damnjanovic2024}.

\begin{lemma}\label{neighbors_lemma}
For any connected graph $G$ and edge $u v \in E(G)$, we have
\[
    \Tr_G(u) - \Tr_G(v) = n_G(v, u) - n_G(u, v).
\]
\end{lemma}

As noted in \cite{AXu}, the next result directly follows from Lemma \ref{neighbors_lemma}.

\begin{corollary}\label{components_cor}
For any TI tree $T$ and vertex $v \in V(T)$, the components of $T - v$ have pairwise distinct orders.
\end{corollary}

We also need the following short lemma.

\begin{lemma} \label{centroid_lemma}
Let $T$ be a TI tree of order $n \ge 2$, with $v \in V(T)$ having the minimum transmission. Then $d_T(v) \ge 3$ and all the components of $T - v$ have fewer than $\frac{n}{2}$ vertices.
\end{lemma}
\begin{proof}
Let $k = d_T(v)$ and observe that $k \ge 2$. Let $T_1, T_2, \ldots, T_k$ be the components of $T - v$. From Lemma \ref{neighbors_lemma}, we get $|T_i| < \frac{n}{2}$ for each $i \in [k]$. Now, by way of contradiction, suppose that $k = 2$. By Corollary \ref{components_cor}, we may assume without loss of generality that $\frac{n}{2} > |T_1| > |T_2|$. Therefore, we obtain
\[
    n = |T| = 1 + |T_1| + |T_2| \le 2 \, |T_1| < n,
\]
which yields a contradiction.
\end{proof}

Recall that an \emph{ordered rooted tree} is a tree with the additional properties that one specific vertex is designated as the root and the children of each vertex are totally ordered in some manner. For any ordered rooted tree $T$, we will use $r(T)$ to denote the root of $T$ and $c_T(v) = (u_1, u_2, \ldots, u_k)$ to denote the sequence of children of a vertex $v \in V(T)$. Note that $c_T(v)$ is an empty sequence when $v$ is a leaf in $T$. Also, for any $v \in V(T), \, v \neq r(T)$, let $p_T(v)$ denote the parent of $v$, and for any $v \in V(T)$, let $\ell_T(v)$ and $T_v$ denote the level of $v$ in $T$ and the rooted subtree of $T$ whose root is $v$, respectively. We will take the level numbering to be zero-based, so that $\ell_T(v) = d_T(v, r(T))$ and the only vertex on level 0 is the root $r(T)$. Besides, let $D(T)$ denote the depth, i.e., the maximum level, of $T$. We now introduce the notions of unbalanced tree and weakly transmission irregular tree as follows.

\begin{definition}
An ordered rooted tree $T$ is \emph{unbalanced} if we have $|T_{u_1}| < |T_{u_2}| < \cdots < |T_{u_k}|$ for every $v \in V(T), \, c_T(v) = (u_1, u_2, \ldots, u_k)$. 
\end{definition}

\begin{definition}
An ordered rooted tree is \emph{weakly transmission irregular}, or WTI for short, if any two distinct vertices from the same level have different transmissions. 
\end{definition}

From Corollary \ref{components_cor}, we observe that every TI tree $T$ has a unique representation as an unbalanced tree whose root is the vertex with the minimum transmission. We will refer to this unbalanced tree as the \emph{canonical representation} of $T$. Note that the canonical representation of $T$ is WTI for any TI tree $T$. We proceed with the following two auxiliary lemmas.

\begin{lemma}\label{transmission_root_lemma}
For any ordered rooted tree $T$, we have
\[
    \Tr_T(r(T)) = \sum_{u \in c_T(r(T))} \Tr_{T_u}(u) + |T| - 1 .
\]
\end{lemma}
\begin{proof}
Let $v = r(T)$. The result follows from
\begin{align*}
    \pushQED{\qed}
    \Tr_T(v) &= \sum_{w \in V(T) \setminus \{v\}} d_T(v, w) = \sum_{u \in c_T(v)} \sum_{w \in V(T_u)} d_T(v, w)\\
    &= \sum_{u \in c_T(v)} \sum_{w \in V(T_u)} \left( d_{T_u}(u, w) + 1 \right) = \sum_{u \in c_T(v)} \left( \Tr_{T_u}(u) + |T_u| \right)\\
    &= \sum_{u \in c_T(v)} \Tr_{T_u}(u) + \sum_{u \in c_T(v)} |T_u| = \sum_{u \in c_T(v)} \Tr_{T_u}(u) + |T| - 1. \qedhere
\end{align*}
\end{proof}

\begin{lemma}\label{transmission_difference_lemma}
Let $T$ be an ordered rooted tree, $n = |T|$, $v \in V(T)$ and $n' = |T_v|$. Also, let $u \in V(T_v), \, u \neq v$, and $u' = p_{T_v}(u) = p_T(u)$. Then we have
\[
    \Tr_T(u) - \Tr_{T_v}(u) = \left( \Tr_T(v) - \Tr_{T_v}(v) \right) + (n - n') \, \ell_{T_v}(u).
\]
\end{lemma}

\begin{proof}
From Lemma~\ref{neighbors_lemma}, we get
\[
    \Tr_{T_v}(u) = \Tr_{T_v}(u') + (n' - 2 \, |T_u|) \quad \mbox{and} \quad \Tr_T(u) = \Tr_T(u') + (n - 2 \, |T_u|),
\]
which implies
\[
    \Tr_T(u) - \Tr_{T_v}(u) = \left( \Tr_T(u') - \Tr_{T_v}(u') \right) + (n - n') .
\]
The result now follows by induction on the $(u, v)$-path.
\end{proof}

The next result follows directly from Lemma \ref{transmission_difference_lemma}.

\begin{corollary} \label{subtree_lemma}
Let $T$ be the canonical representation of a TI tree. Then $T_v$ is WTI for every $v \in V(T)$.
\end{corollary}
\begin{proof}
If two distinct vertices $u_1, u_2 \in V(T_v)$ such that $\ell_{T_v}(u_1) = \ell_{T_v}(u_2)$ have different transmissions in $T$, then Lemma \ref{transmission_difference_lemma} implies that they also have different transmissions in $T_v$. Since $T$ is WTI, we conclude that $T_v$ is also WTI.
\end{proof}

\section{Algorithm overview}\label{sc_algorithm}

In this section, we describe an algorithm that generates all TI trees of order at most $n$ and with maximum degree at most $m$, where $n, m \in \mathbb{N}$ are given as arguments.
The algorithm operates on WTI trees and filters the results for TI trees as the final step.
The core component of the algorithm is a function that combines one or more WTI trees $T_1, T_2, \ldots, T_k$ with strictly increasing orders into one WTI tree $T$ by introducing a new root vertex and connecting it to the root vertices of $T_1, T_2, \ldots, T_k$.
Clearly, when joining several such trees, it is possible that the resulting tree is not a WTI tree. In that case, the function reports failure.
Subsection \ref{wti_trees_subsection} describes this function in detail, along with the data structure used to store information on WTI trees. Subsection \ref{wti_algo} describes a function that generates all WTI trees up to a given order, while Subsection \ref{overall_algo} elaborates the strategy employed to efficiently generate TI trees using the above functions.

\subsection{Manipulating WTI trees} \label{wti_trees_subsection}

By convention, the root vertex of each WTI tree $T$ is labeled $0$, while the other vertices are $1, 2, \ldots, n - 1$, where $n = |T|$. For each WTI tree $T$, the algorithm stores the minimum information needed to efficiently perform the joining procedure:
\begin{enumerate}[label=\textbf{(\roman*)}]
    \item the order, i.e., $|T|$; 
    \item the number of levels, i.e., $D(T) + 1$;
    \item the parent array $p_0, p_1, \ldots, p_{n-1}$, where $p_i$ is the label of $p_T(i)$, the parent of vertex~$i$, while $p_0$ is left undefined;
    \item for each level $i \in \{ 0, 1, \ldots, D(T) \}$, the list $L_i$ of transmission values for all the vertices from level $i$.
\end{enumerate}

Given WTI trees $T_1, T_2, \ldots, T_k$, we first compute the order and number of levels of $T$ by using the formulae
\[
    |T| = 1 + \sum_{i=1}^k |T_i| \qquad \mbox{and} \qquad D(T) = 1 + \max_{i = 1}^k D(T_i).
\]
Afterwards, for all $i \in [k]$, the vertices from the tree $T_i$ are assigned new labels in $T$. If a vertex has the label $x$ in $T_i$, then in $T$, it is given the label $x + 1 + \sum_{j=1}^{i-1}|T_j|$. The root vertex of $T$ is, of course, given the label $0$. We trivially observe that this procedure assigns unique labels from $0, 1, \ldots, |T| - 1$ to all the vertices of $T$.

We apply Lemma \ref{transmission_root_lemma} to compute $\Tr_T(r(T))$. Afterwards, $\Tr_T(u)$ is computed for every $u \in c_T(r(T))$ via Lemma \ref{neighbors_lemma}. Finally, for each $i \in [k]$, we compute $\Tr_T(u)$ for every $u \in V(T_i)$ by using Lemma \ref{transmission_difference_lemma}. Each computed transmission from level $i$ of $T$ is accordingly stored in the corresponding list $L_i$. The function reports failure if any of the obtained lists $L_i$ does not contain distinct values.

\subsection{Generating WTI trees} \label{wti_algo}

The function \texttt{generateWTITrees} that generates WTI trees takes two arguments:\ $n \in \mathbb{N}$, the upper bound on the tree order, and $h \in \mathbb{N}$, the upper bound on the number of children that a tree vertex can have. Denote the function output by \texttt{result}. This value is an array of collections of WTI trees, as described in Subsection \ref{wti_trees_subsection}, where \texttt{result[k]} contains all the WTI trees of order $k$ whose vertices have at most $h$ children.

The function uses two auxiliary functions. These functions work in-place and output their results by invoking a provided callback once for each resulting item:
\begin{enumerate}[label=\textbf{(\roman*)}]
    \item \texttt{generateIncreasing}:\ generates all the strictly increasing sequences $(s_1, s_2, \linebreak \ldots, s_q) \in \mathbb{N}^q$ such that $\sum_{i=1}^q s_i = \alpha$, $s_q \le \beta$ and $q \le \gamma$, where $\alpha, \beta, \gamma \in \mathbb{N}$ are given parameters;
    \item \texttt{cartesianProduct}:\ computes the cartesian product of a given sequence of collections of items.
\end{enumerate}

The collection \texttt{result[1]} is initialized with a single tree --- the unique tree of order $1$, which is a WTI tree. The trees of orders $k \in \{2, 3, \ldots, n\}$ are then produced, in that order, by computing the cartesian product of the sequence
\[
    (\mathtt{result[s_1]}, \mathtt{result[s_2]}, \ldots, \mathtt{result[s_q]})
\]
for each strictly increasing sequence $(s_1, s_2, \ldots, s_q) \in \mathbb{N}^q$ such that $\sum_{i=1}^q s_i = k - 1$ and $q \le h$. These sequences are, in turn, generated by the function \texttt{generateIncreasing}. Each tuple of such a cartesian product is then combined using the procedure described in Subsection \ref{wti_trees_subsection} and stored in \texttt{result[k]} if the joining procedure was successful.

\subsection{Generating TI trees} \label{overall_algo}

The function \texttt{generateTITrees} that generates TI trees takes three arguments:\ $n \in \mathbb{N}$, the upper bound on the tree order, $m \in \mathbb{N}$, the upper bound on the tree maximum degree, and a function \texttt{func} that will be called once for each generated TI tree. This approach, inspired by the functional programming paradigm, augments the flexibility of the code and improves performance, as it does not require all TI trees to be held in memory at once.

The function first calls \texttt{generateWTITrees} to generate all the WTI trees of order at most $\nhalf$ such that no tree vertex has more than $m$ children. Note that some of the generated trees may violate the constraint on the maximum degree of a vertex. This happens when a generated tree has a non-root vertex with $m$ children. Such a vertex will thus have a degree of $m + 1$ and, therefore, the tree must be discarded.

Two auxiliary functions are introduced:
\begin{enumerate}[label=\textbf{(\roman*)}]
    \item \texttt{getMaxDegree}:\ finds the maximum degree and the number of children of the root vertex of a given WTI tree;
    \item \texttt{isTITree}:\ inspects whether the given WTI tree is the canonical representation of a TI tree. This is done by checking the list of transmission values for uniqueness and by checking if the minimum value is in the root vertex.
\end{enumerate}

For each tree generated by \texttt{generateWTITrees}, if the maximum degree exceeds $m$, then the tree is discarded.
For all the remaining trees, we check if it is the canonical representation of a TI tree, and in the affirmative case, this is reported by invoking the callback \texttt{func}. In addition, regardless of whether we get the canonical representation of a TI tree, if the tree root has strictly fewer than $m$ children, it is stored in the auxiliary collection \texttt{subtrees[k]}, where $k$ is the tree order.

Until now, we have described how the function generates the TI trees with orders up to $\nhalf$. To generate the trees of orders $k \in \{\nhalf + 1, \ldots,n-1,n\}$, we rely on Lemma \ref{centroid_lemma}. We directly generate WTI trees of order $k$ whose roots have the minimum transmission by placing an upper bound on the order of subtrees to $\khalf$. We invoke \texttt{generateIncreasing} by setting $\alpha = k-1$, $\beta = \khalf$ and $\gamma = m$. Then, for each generated sequence $(s_1, s_2, \ldots, s_q) \in \mathbb{N}^q$, we take the cartesian product of collections $\mathtt{subtree[s_1]}, \mathtt{subtree[s_2]}, \ldots, \mathtt{subtree[s_q]}$, join the trees from each tuple and check if the joined tree is a TI tree. If so, the tree is outputted by calling \texttt{func}.

The function \texttt{generateTITrees} is easily parallelized by executing the generation of each cartesian product as a separate asynchronous task. Additional optimizations done in the \texttt{C++20} code include the use of fixed-size arrays for representing WTI trees, using minimum width integer types and using a custom implementation of arrays of arrays whose layout is known upon construction.

The main function of the program takes two mandatory arguments:\ the mode and the maximum tree order, and one optional argument:\ the maximum vertex degree. The mode can be one of the following:
\begin{enumerate}[label=\textbf{(\roman*)}]
    \item \texttt{-c}:\ the program prints the number of trees of each order;
    \item \texttt{-g}:\ the program prints all the trees in the \texttt{graph6} format;
    \item \texttt{-s}:\ the program prints all the trees in the \texttt{sparse6} format;
    \item \texttt{-p}:\ for every tree, the program prints the parent of each non-root vertex.
\end{enumerate}

\section{Execution results}\label{sc_results}

\begin{table}[t]
\centering
\begin{tabular}{r|r||r|r||r|r}
$n$ & $\mathcal{TI}_n$ & $n$ & $\mathcal{TI}_n$ & $n$ & $\mathcal{TI}_n$\\
\hline\hline
1 & 1 & 15 & 82 & 29 & 3,277,565 \\
2 & 0 & 16 & 10 & 30 & 302,163 \\
3 & 0 & 17 & 324 & 31 & 16,926,170 \\
4 & 0 & 18 & 47 & 32 & 1,368,434 \\
5 & 0 & 19 & 1,574 & 33 & 83,965,665 \\
6 & 0 & 20 & 165 & 34 & 7,612,216 \\
7 & 1 & 21 & 6,944 & 35 & 409,768,230 \\
8 & 0 & 22 & 733 & 36 & 33,750,452 \\
9 & 1 & 23 & 30,913 & 37 & 2,140,396,213 \\
10 & 0 & 24 & 2,947 & 38 & 162,433,861 \\
11 & 6 & 25 & 143,690 & 39 & 11,148,732,100 \\
12 & 0 & 26 & 13,357 & 40 & 824,703,171 \\
13 & 24 & 27 & 702,945 & 41 & 56,072,411,946 \\
14 & 1 & 28 & 67,685 & 42 & 4,343,163,569 \\
\end{tabular}
\caption{$\mathcal{TI}_n$ denotes the number of nonisomorphic TI trees of order $n$.}
\label{brojke}
\end{table}

The developed program was tested on a desktop computer with a modern, high-end CPU. Generating and printing all the TI trees of order up to $26$ takes under $0.4$ seconds, while the program \texttt{gentreeg} from the package \texttt{nauty} takes approximately $22.6$ seconds to generate all the nonisomorphic trees of order at most $26$. The difference becomes even more pronounced when considering larger trees.

The program was used to count the number of nonisomorphic TI trees of all orders up to $42$ and the execution took less than $12$ hours. On the other hand, we estimate that \texttt{gentreeg} could take at least $7$ years just to generate the trees, not counting the time needed to select the TI trees among them. The number of nonisomorphic TI trees of each order $n \le 42$ is given in Table \ref{brojke}. Note that the algorithm takes the trivial tree to be TI, so that $\mathcal{TI}_1 = 1$.

\section*{Acknowledgements}

I.\ Damnjanović is supported in part by the Science Fund of the Republic of Serbia, grant \#6767, Lazy walk counts and spectral radius of threshold graphs --- LZWK.

\section*{Conflict of interest}

The authors declare that they have no conflict of interest.

\end{document}